\documentclass[letterpaper, 10 pt, conference]{ieeeconf}  

\IEEEoverridecommandlockouts                              

\usepackage{amsmath, amssymb, amsthm} 
\usepackage{graphics} 
\usepackage[nocomma, short]{optidef} 
\usepackage{multicol}
\usepackage{booktabs} 
\usepackage{tikz} 
\usetikzlibrary{math}
\usetikzlibrary{arrows.meta}
\usepackage{float}
\usepackage{hyperref} 
\usepackage{longtable} 
\usepackage{adjustbox}
\usepackage{pifont}
\newcommand{\cmark}{\ding{51}}%
\newcommand{\xmark}{\ding{55}}%

\newtheorem{theorem}{Theorem}[section]

\newtheorem{lemma}[theorem]{Lemma}

\newcommand{\ra}[1]{\renewcommand{\arraystretch}{#1}} 

\newcommand{\set}[1]{\mathcal{ #1 }}

\newcommand{\ub}[1]{\overline{ #1 }}
\newcommand{\lb}[1]{\underline{ #1 }}

\newcommand{\R}{\mathbb{R}}

\newcommand{\norm}[1]{\left\lVert #1 \right\rVert}

\newcommand{\control}{u}

\newcommand{\frenetTrans}{\mathcal{F_\gamma}}

\newcommand{\free}{\set{P}}

\newcommand\copyrighttext{%
	\footnotesize \textcopyright 2022 IEEE. Personal use of this material is permitted. Permission from IEEE must be obtained for all other uses, in any current or future media, including reprinting/republishing this material for advertising or promotional purposes, creating new collective works, for resale or redistribution to servers or lists, or reuse of any copyrighted component of this work in other works.}
\newcommand\copyrightnotice{%
	\begin{tikzpicture}[remember picture,overlay]
		\node[anchor=south,yshift=5pt] at (current page.south) {\fbox{\parbox{\dimexpr\textwidth-\fboxsep-\fboxrule\relax}{\copyrighttext}}};
	\end{tikzpicture}%
}

\DeclareMathSymbol{\shortminus}{\mathbin}{AMSa}{"39}

\title{\LARGE \bf
	Frenet-Cartesian Model Representations for Automotive Obstacle Avoidance within Nonlinear MPC 
}

\author{Rudolf Reiter$^{1}$, Armin Nurkanovi\'c$^{1}$, Jonathan Frey$^{1,2}$ and Moritz Diehl$^{1,2}$
	\thanks{$^{1}$Department of Microsystems Engineering (IMTEK), University Freiburg, 79110 Freiburg, Germany
	{\tt\small \{rudolf.reiter, armin.nurkanovic, jonathan.frey, moritz.diehl\}@imtek.uni-freiburg.de}}%
	\thanks{$^{2}$Department of Mathematics, University Freiburg, 79110 Freiburg, Germany
	}%
}

\begin{document}
	
	\maketitle
	\copyrightnotice
	\thispagestyle{empty}
	\pagestyle{empty}

	\begin{abstract}
		In recent years, nonlinear model predictive control (NMPC) has been extensively used for solving automotive motion control and planning tasks.
		In order to formulate the NMPC problem, different coordinate systems can be used with different advantages.
		We propose and compare formulations for the NMPC related optimization problem, involving a Cartesian and a Frenet coordinate frame (CCF/ FCF) in a single nonlinear program (NLP).
		We specify costs and collision avoidance constraints in the more advantageous coordinate frame, derive appropriate formulations and compare different obstacle constraints.
		With this approach, we  exploit the simpler formulation of opponent vehicle constraints in the CCF, as well as road aligned costs and constraints related to the FCF.
		Comparisons to other approaches in a simulation framework highlight the advantages of the proposed approaches.
	\end{abstract}
	\section{Introduction}
	\label{section:introduction}
		Trajectory optimization with obstacle avoidance is a major challenge of motion planning and control in autonomous driving and racing. 
		Trajectories need to be feasible to kinodynamic equations and avoid collisions with objects that are often hard to predict.
		Collision avoidance and the related generation of a reference trajectory or collision avoidance as part of the controller, e.g., model predictive control (MPC), is often formulated as a nonlinear discrete time-optimal optimal control problem \cite{Kloeser2020, Raji2022, Liniger2015, Brito2019}.
		Through a carefully chosen nonlinear program (NLP) formulation and by means of dedicated real-time optimization solvers \cite{Verschueren2021, Sathya2018}, the problem can be solved efficiently.
		The transformation of the dynamics into a road aligned coordinate frame (CF), namely the Frenet CF (FCF), has shown many advantages, such as the simplification of references and road boundaries \cite{Kloeser2020,Reiter2021,Werling2010}.
		Nevertheless, the transformed coordinates also come with the disadvantage of transformed geometric obstacle shapes \cite{Xing2022}, cf. Sec.~\ref{sec:model_comparison}.
		Typical convex geometric shapes, such as boxes, ellipses or circles, are easier to describe in the Cartesian reference CF and become nonconvex after transformation into the FCF.
		The shapes of objects in both frames are shown in Fig. \ref{fig:race_car}.
		\begin{figure}
			\begin{center}
				\includegraphics[scale=0.85]{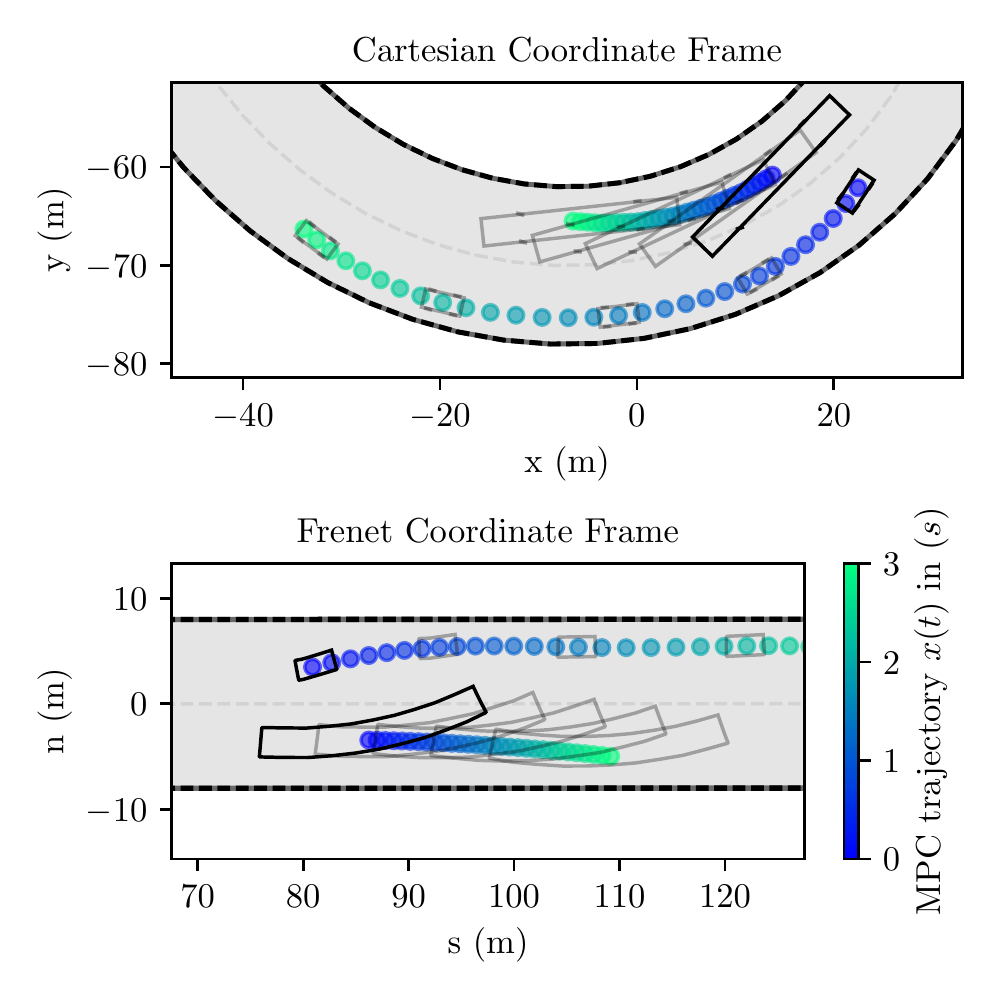}
				\caption{Simulated overtaking of the same maneuver in two CFs, namely the Cartesian (top plot) and Frenet CF (bottom plot). 
					Planned trajectories with~$\Delta t=0.1$s and snapshots every~$0.7$s of boundary box alignments.}
				\label{fig:race_car}
			\end{center}
		\end{figure}
		In nonlinear optimization, "lifting" is a technique where the optimization problem is formulated and solved in a higher dimensional space, which offers advantages regarding convergence rates and region of attraction \cite{Albersmeyer2010}.
		The contribution of the presented paper is a way to extend and lift the state space of the vehicle model by including both CFs and thereafter formulate constraints and costs in the more appropriate CF.
		We show an increase of overall performance due to the improved description of the obstacle shapes with a variety of deterministic obstacle avoidance formulations in simulation.
		Despite the increased state dimensions, even the computation time can be lowered compared to a pure FCF representation.
		Additionally, references can be set in any of the two CFs, which allows for flexible combination with planning modules that use either CF, e.g., \cite{Varquez2020}.\\
		\textbf{Related work:}
			The effectiveness of NMPC using the FCF related to automotive tasks was shown in numerous works \cite{Kloeser2020, Raji2022, Reiter2021,Werling2010, Varquez2020, Wang2021, Buyval2017, Reiter2022, Li2021, Rosolia2017, Ayoub2022, Wang2015}.
			None of them explicitly considers the shape transformation of objects, which are rather over-approximated with convex shapes in the FCF.
			Convex obstacle shapes in the Cartesian CF (CCF) are considered in \cite{Rasekhipour2016} with potential fields, in \cite{Wang2021,Ziegler2014,Jung2021} with covering circles and Euclidean distance constraints, in \cite{Brito2019, Nair2022a} with ellipses, in \cite{Brossette2017, Nair2022} with separating hyperplanes and \cite{Sathya2018, Evens2022} with a formulation related to a conjunction of convex planes covering the obstacle.
			The most prominent variants are compared within this paper in the Frenet and the lifted formulation.
			More importantly, the shape fitting problem with transformed objects in the FCF and an approach with surrogate representations in both CFs was recently considered in a related way in~\cite{Xing2022}.
			However, \cite{Xing2022} focuses on different topics, as it deals with linear MPC and does not consider dedicated obstacle formulations.
			Furthermore, it integrates an approximation of the transformation itself into the model, i.e., and approximation of a DAE, whereas in our formulation we provide a reduced index formulation which constitutes an ODE (see Sec.~\ref{sec:overview_lifting}) or eliminate algebraic variables directly. 
			Another variant of tracking along a reference path stems from \cite{Lam2010} and was extended to vehicles in \cite{Liniger2015} and \cite{Kuchera2013}.
			It uses a method called \emph{contouring control}, which uses a state on a path-length parameterized reference curve and an additional state for its path position.
			Similarly to \cite{Xing2022}, it approximates the transformation implicitly, which involves the approximation of a bi-level optimization problem for finding the closest point on the reference curve.\\
			\textbf{Contribution}:
			This paper proposes novel NMPC formulations that extend the state space to two CFs and allow for an efficient handling of the occurring costs and constraints.
			Thereby, the usually convex and simple obstacle shapes in the Cartesian CF can be directly used in the NMPC formulation.
			We show in simulation that we outperform the conventional approach of over-approximation \cite{Li2021, Rosolia2017, Ayoub2022, Wang2015} in terms of computation time and performance.
			Furthermore, the obstacle shapes are independent of the states and the road (up to Euclidean transformations), which is not the case in a conventional Frenet representation.
			Additionally, we contribute with an extensive comparison of common obstacle avoidance formulations in the proposed formulations.
	\section{Vehicle Models}
	\label{section:models}
		In order to formulate our NMPC problem, we use a rear-axis referenced kinematic vehicle model \cite{Kloeser2020, Reiter2021}.
		It comprises three states that are related to the CF which we write as~$x^\mathrm{c}$. Particularly, we use~$x^\mathrm{c,C}=[x\quad y\quad\varphi]^\top \in \R^3$ for the Cartesian states and~$x^\mathrm{c,F}=[s\quad n\quad\alpha]^\top\in \R^3$ for the Frenet states.
		We use the Cartesian (earth) position states~$x$,~$y$ and the heading angle~$\varphi$.
		Similarly, in the FCF we use position states~$s$ and~$n$, together with the difference angle~$\alpha$, which are all related to a reference curve.
		The other states~$x^\mathrm{\neg c}=[v\quad\delta]^\top\in\R^2$ are needed for both, CCF and FCF, where~$v$ is the absolute value of the velocity at the rear axis and~$\delta$ is the steering angle.
		For the full CCF model we use the state $x^\mathrm{C}=[x^{\mathrm{c,C}\top} \quad x^{\mathrm{\neg c}\top}]^\top$ and for the FCF model we use the state $x^\mathrm{F}=[x^{\mathrm{c,F}\top} \quad x^{\mathrm{\neg c}\top}]^\top$.
		We assume a rear-wheel drive force~$F^\mathrm{d}$ as input, which includes the acceleration and braking force.
		We simplify that the braking force of the front wheel acts in the direction of the rear-wheel, thus we include it in the force~$F^\mathrm{d}$. 
		This approximation is valid for small steering angles, since the projection of the front wheel force~$F^\mathrm{d,front}$ on the rear-wheel reference is given by~$F^\mathrm{d,front}\cos(\delta)$.
		The most prominent resistance forces for wind $F^\mathrm{wind}(v,\varphi)=c_\mathrm{air}v_\mathrm{rel}(v,\varphi)^2$ and the rolling resistance $F^\mathrm{roll}=c_\mathrm{roll}\mathrm{sign}(v)$ are included, with the total resistance force~$F^\mathrm{res}(v,\varphi)=F^\mathrm{wind}(v,\varphi)+c_\mathrm{roll}\mathrm{sign}(v)$. 
		The air drag depends on the vehicle speed~$v$ in relation to the wind speed~$v_\mathrm{wind}$ with the air friction parameter~$c_\mathrm{air}$.
		The rolling resistance is proportional to~$\mathrm{sign}(v)$ by the constant~$c_\mathrm{roll}$.
		We drop the sign function, since we only consider strictly positive speeds.
		We model the relative speed related to the air drag, which we assume constant and known, by~$v_\mathrm{rel}(v,\varphi)=v+v_\mathrm{wind}\cos(\varphi-\varphi_\mathrm{wind})$, where~$\varphi_\mathrm{wind}$ is the angle of the direction the wind asserts force and~$\varphi$ is the heading of the vehicle in the CCF. 
		In most works (e.g., \cite{Liniger2015,Reiter2021}) the influence of external wind is ignored, since its influence might be small, unpredictable or absent in experimental indoor setups.
		Nevertheless, we include it, since it demonstrates an influence that can be easily modeled in the FCF, but is difficult to model in the FCF.\\
		The input of our model is given by $u=[F^\mathrm{d}\quad r]^\top\in\R^2$, where~$r = \frac{\mathrm{d} \delta}{ \mathrm{d} t}$ denotes the steering rate.
		The dynamics of the coordinate unrelated states are written as
		\begin{equation}
			\dot{x}^\mathrm{\neg c}=
			f^\mathrm{\neg c}(x^\mathrm{\neg c},u,\varphi)=
			\begin{bmatrix}
			\frac{1}{m}(F^\mathrm{d}+F^\mathrm{wind}(v,\varphi)-F^\mathrm{res}(v))\\
			r
			\end{bmatrix},
		\end{equation}
		where $m$ denotes the vehicle mass.
		The lateral acceleration $a_\mathrm{lat}(x^\mathrm{\neg c})$ at the rear wheel axis is given by
		\begin{equation}
			a_\mathrm{lat}(x^\mathrm{\neg c})=\frac{v^2\tan(\delta)}{l},
		\end{equation}
		where~$l=l_\mathrm{f}+l_\mathrm{r}$ is the total wheelbase length of the vehicle.
		\subsection{CCF vehicle model}
			Using simple kinematic relations, the dynamics of the Cartesian states can be written as
			\begin{equation}
			\label{eq:model_cartesian_only}
				\dot{x}^\mathrm{c,C}=
				f^\mathrm{c,C}(x^\mathrm{C}, u)=
				\begin{bmatrix}
				v \cos(\varphi)\\
				v \sin(\varphi)\\
				\frac{v}{l}\tan(\delta)
				\end{bmatrix}.
			\end{equation}
			The full five-state Cartesian vehicle model is given by
			\begin{equation}
			\label{eq:model_cartesian_full}
				\dot{x}^\mathrm{C}=
				\begin{bmatrix}
				f^\mathrm{c,C}(x^\mathrm{C}, u)\\
				f^\mathrm{\neg c}(x^\mathrm{\neg c},u,\varphi)
				\end{bmatrix}.
			\end{equation}
		\subsection{FCF vehicle model}
			Since in usual vehicle motion control tasks, the vehicle moves mainly close to a reference curve~${\gamma:\R \rightarrow \R^2}$, i.e., the street center line, the transformation into a curvilinear CF is a natural choice.
			The reference curve~$\gamma(\sigma)=[\gamma_x(\sigma) \quad \gamma_y(\sigma)]^\top$ is parameterized by its path length~$\sigma$ and can be fully described by one initial transformation offset~$\gamma(\sigma_0)$, an initial orientation~$\varphi_0$ and the curvature~${\kappa(\sigma)=\frac{\mathrm{d}\varphi^\gamma}{\mathrm{d}\sigma}}$ along its path.
			We use~$\varphi^\gamma(\sigma)$ for the tangent angle in each point of the curve.
			As part of the Frenet transformation, we project the Cartesian vehicle reference point~$p^\mathrm{veh}\in\R^2$ on the closest point in the Euclidean distance of the reference curve with
			\begin{equation}
			\label{eq:frenet_transformation_argmin}
			s^*(p^\mathrm{veh})=\arg \min_{\sigma}\norm{p^\mathrm{veh}-\gamma(\sigma)}^2_2 .
			\end{equation}
			W.l.o.g., we always set the initial reference point of the transformation to zero.
			The position~$s$ on the curve is used as longitudinal FCF state.
			The vector~$(p^\mathrm{veh}-\gamma(s^*))$ is difference of the closest point on the curve to the vehicle.
			By using the 90 degree rotation matrix~$R^{90}$ and projection to the normal unit vector of the curve~$e_n=R^{90}\gamma'(s^*)$, we obtain the Frenet state~$n=(p^\mathrm{veh}-\gamma(s^*))^\top e_n$.
			The third Frenet state~$\alpha$ relates the tangent angle of the curve to the heading of the vehicle with~$\alpha=\varphi-\varphi^\gamma(s^*)$.
			The relations of the transformation are shown in Fig. \ref{fig:frame_vehicle_model}.
			We write the transformation of a Cartesian state~$x^\mathrm{c,C}=[x\quad y\quad \varphi]^\top$ to a Frenet state~$x^\mathrm{c,F}=[s\quad n\quad \alpha]^\top$ by means of the transformation 
			\begin{align}
				x^\mathrm{c,F}=\frenetTrans(x^\mathrm{c,C})=\begin{bmatrix}
				s^* \\
				(p^\mathrm{veh}-\gamma(s^*))^\top e_n \\
				\varphi^\gamma(s^*)-\varphi 
				\end{bmatrix},
			\end{align}
			and its inverse by
			\begin{align}
				x^\mathrm{c,C}=\frenetTrans^{-1}(x^\mathrm{c,F})=\begin{bmatrix}
				\gamma_x(s)-n\sin(\varphi^\gamma(s)) \\
				\gamma_y(s)+n\cos(\varphi^\gamma(s)) \\
				\varphi^\gamma(s)-\alpha 
				\end{bmatrix}.
			\end{align}
			The existence and uniqueness of the transformation is guaranteed under mild assumptions which are discussed in detail in \cite{Reiter2021}. 
			\begin{figure}
				\begin{center}
					\def\svgwidth{.4\textwidth}
					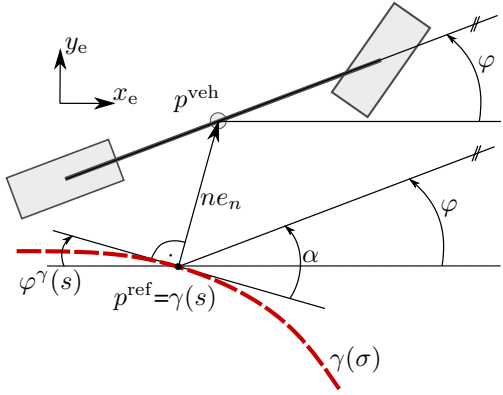
					\caption{State relations between Cartesian and FCF}
					\label{fig:frame_vehicle_model}
				\end{center}
			\end{figure}
			As shown in \cite{Kloeser2020}, we obtain the ODE for the kinematic motion in the FCF as
			\begin{equation}
			\label{eq:model_frenet_only}
				\dot{x}^\mathrm{c,F}=
				f^\mathrm{c,F}(x^\mathrm{F}, u)=
				\begin{bmatrix}
					\frac{v \cos(\alpha)}{1-n \kappa(s)}\\
					v \sin(\alpha)\\
					\frac{v}{l}\tan(\delta) - \frac{\kappa(s) v \cos(\alpha)}{1-n \kappa(s)}
				\end{bmatrix}.
			\end{equation}
			The Cartesian state~$\varphi$ is needed in order to formulate the wind disturbance.
			It is not available in the FCF, consequently it needs to be computed by evaluating the tangent angle~$\varphi^\gamma(s)$ of the current position~$s$ on the reference curve~$\gamma(\sigma)$.
			This can be approximated by a spline function~$\hat{\varphi}^\gamma(s)$ that is computed for the road layout.
			It yields a spline approximation for the heading angle with~$\hat{\varphi}(s,\alpha)=\hat{\varphi}^\gamma(s)+\alpha$.
			The full FCF vehicle model is consequently given by the five state model
			\begin{equation}
			\label{eq:model_frenet_full}
				\dot{x}^\mathrm{F}=
				\begin{bmatrix}
				f^\mathrm{c,F}(x^\mathrm{F}, u)\\
				f^\mathrm{\neg c}(x^\mathrm{\neg c},u,\hat{\varphi})
				\end{bmatrix}.
			\end{equation}
		\subsection{Model comparison}
		\label{sec:model_comparison}
			As indicated in Sec.~\ref{section:introduction}, the CF models have different advantages when used in a NMPC formulation, see Tab.~\ref{table:model_comparison} for an overview.
			The definition of road boundaries and the reference curve, which are often center-lane-aligned curves are straight forward to define in the FCF, but hard to define in the CCF.
			However, the obstacle definition in the FCF is cumbersome due to several reasons.
			Despite convex obstacle shapes in the CCF, safety cannot be guaranteed when using SQP to solve the NMPC problem with the Frenet model.
			Convex obstacle shapes cannot be guaranteed to be convex, if transformed into the FCF.
			This fact can be seen from the following example.
			Consider a straight line, which is a convex set, and a circular road. Let the line intersect the road at coordinates~$\gamma(\sigma_1)=[x_1\quad y_1]^\top$ and~$\gamma(\sigma_2)=[x_1\quad y_1]^\top$. The transformed Frenet states~$n_1,n_2$ are zero in either point. At~$\sigma_3\in(\sigma_1,\sigma_2)$ the transformed state~$n_3\neq 0$, thus the transformed set is not convex.
			By considering Lemma \ref{lemma:convex_safe} it can be shown that convex obstacles are guaranteed to be a subset of the linearized constraints within an SQP iteration, thus safely over-approximated.
			\begin{lemma}
			\label{lemma:convex_safe}
				Regard the set~$\mathcal{C} = \{x\in \R^{n} \mid g(x) \geq 0\}$  and $\mathcal{C}^{\mathrm{lin}}(x^*) = \{x\in \R^{n} \mid g(x^*) + \nabla g(x^*)^\top (x-x^*) \geq 0\}$. Suppose that the function~$g:\R^n \rightarrow \R$ is convex, then $ \mathcal{C} \subseteq \mathcal{C}^{\mathrm{lin}}(x^*)$ for any $x^*$.
			\end{lemma}
			\begin{proof}
				Due to convexity, $g(x^*)+\nabla g(x^*)^\top(x-x^*)\leq g(x)$ and therefore, it follows that $\mathcal{C} \subseteq \mathcal{C}^\mathrm{lin}$.
			\end{proof}
			Nonconvex obstacles are not safely over-approximated within SQP algorithms.
			Another problem that arises with objects in the FCF is the dependence of the shape on the state.
			Consequently, if the obstacle constraints are defined along a discretized time horizon, at each time step $i=0,\ldots,N$, the shape has to be transformed separately, cf. Fig.\ref{fig:race_car}.
			In typical applications this could lead to $N$ transformations for each obstacle in every NMPC iteration, followed by a convexification (e.g., bounding boxes, convex polygons, covering circles) in order to guarantee safety.
			Alternatively, an over-approximation could be used to capture all possible transformed shapes.
			However, this would lead to a striking conservatism, especially for long vehicles and low curve radii. 
			\begin{table}
				\centering
				\ra{1.2}
				\begin{tabular}{@{}lcc@{}}
					\addlinespace
					\toprule
					Feature & CCF& FCF\\
					\midrule
					reference definition & \xmark&\cmark\\
					boundary constraints & \xmark&\cmark\\
					obstacle specification & \cmark & \xmark\\
					disturbance specification & \cmark & \xmark\\
					\bottomrule
				\end{tabular}
				\caption{Comparison of the two model representations}
				\label{table:model_comparison}
			\end{table}
		\subsection{Overview of CF lifting formulations}
		\label{sec:overview_lifting}
			As outlined in Sec.~\ref{sec:model_comparison}, it is beneficial to have states of both CFs in the NLP formulation in order to simplify the constraints. 
			Several different ways of including both CFs are possible and a summary is given in Tab.\ref{table:comparison_formulations}.
			\begin{table*}
				\centering
				\ra{1.2}
				\begin{tabular}{@{}llllllll@{}}
					\addlinespace
					\toprule
					Formulation & \begin{tabular}{@{}l@{}}ODE \\ CF\end{tabular} & \begin{tabular}{@{}l@{}}Obstacle \\ CF\end{tabular} & \begin{tabular}{@{}l@{}}Cost \\ CF\end{tabular} & $n_x$ & $n_z$ & \begin{tabular}{@{}l@{}}Practical \\ Relevance\end{tabular}  &Issues \\
					\midrule
					\textbf{Conventional Frenet}			& Frenet & Frenet &Frenet& 5 & 0 &  yes & nonconvex state-depenent obstacle shapes\\
					&&&&&&&usually over-approximated \cite{Li2021, Rosolia2017, Ayoub2022, Wang2015}\\
					\textbf{Direct Elimination Frenet }	& Frenet & Cartesian &Frenet& 5 & 0 & yes & additional nonlinearities (objective, constraints)\\
					\textbf{Lifted ODE Frenet} 				& Frenet & Cartesian&Frenet & 8 & 0 & yes & redundant states\\
					DAE Frenet 				& Frenet & Cartesian&Frenet & 5 & 3 & no & bad convergence in our experiments\\
					Conventional Cartesian			& Cartesian & Cartesian& Cartesian & 5 & 0 & no & nonconvex boundary constraints \cite{Liniger2015}\\
					Cartesian with Frenet States	& Cartesian & Cartesian& Frenet &$\{5,8\}$ & $\{0,3\}$ & yes & Difficult bi-level problem. Approximations, e.g. \cite{Liniger2015}\\
					\bottomrule
				\end{tabular}
				\caption{Comparison of CF Formulations in NMPC. Bold typed formulations are compared in this paper.}
				\label{table:comparison_formulations}
			\end{table*}
			First, we can choose the main CF ODE and introduce the states related to the other CF as algebraic variables that are determined by the main CF and obtain a differential algebraic equation (DAE) of index~1.
			We could either have a CCF based DAE
			\begin{align}
			\label{eq:dae_cartesian}
			\begin{split}
				\dot{x}^\mathrm{C}&=f^\mathrm{C}(x^\mathrm{C}, u)\\
				0&=x^\mathrm{c,F}-\frenetTrans(x^\mathrm{c,C})
			\end{split}
			\end{align}
			or a FCF based DAE
			\begin{align}
			\label{eq:dae_frenet}
			\begin{split}
			\dot{x}^\mathrm{F}&=f^\mathrm{F}(x^\mathrm{F},x^\mathrm{c,C}, u)\\
			0&=x^\mathrm{c,C}-\frenetTrans^{-1}(x^\mathrm{c,F}).
			\end{split}
			\end{align}
			The inverse transformation~$\frenetTrans^{-1}$ is computationally cheap, since it just needs explicit function evaluations, whereas the forward transformation~$\frenetTrans$ requires solving an NLP as in \eqref{eq:frenet_transformation_argmin}, resulting in a computationally expensive bi-level problem in the final MPC formulation.
			Therefore, we choose the FCF of \eqref{eq:dae_frenet} as a basis and exclude CCF formulations \eqref{eq:dae_cartesian} from further comparisons.
			The DAE of index~1 (\emph{Lifted DAE Frenet}) is a possible way to formulate the problem and was similarly used in \cite{Xing2022} for a linearized model.
			Another possible formulation (\emph{Direct Elimination Frenet}) is to directly eliminate the algebraic variables in \eqref{eq:dae_frenet} by using the inverse Frenet transformation in the nonlinear constraint formulation.
			If the objective includes Cartesian states with quadratic costs and lifted constraints, the direct elimination would lead to a nonlinear objective and constraints.
			Alternatively, we can perform an index reduction of \eqref{eq:dae_frenet}, which is obtained by differentiation of the algebraic constraint, leading to
			\begin{subequations}
			\label{eq:dae_ode_frenet}
			\begin{align}
			\dot{x}^\mathrm{c,F}&=f^\mathrm{c,F}(x^\mathrm{c,F}, u)\\
			\label{eq:dae_frenet_line1}
			0&=\dot{x}^\mathrm{c,C}-\frac{\partial\frenetTrans^{-1}(x^\mathrm{c,F})}{\partial x^\mathrm{c,F}}f^\mathrm{c,F}(x^\mathrm{c,F}, u), \\
			&\quad \text{and}\quad x^\mathrm{c,C}(0):=\frenetTrans^{-1}(x^\mathrm{c,F}).
			\end{align}
			\end{subequations}
			Detailed computation (not presented here) shows the equivalence of \eqref{eq:dae_frenet_line1} to 
			\begin{align}
			\label{eq:dae_ode_frenet_final}
			\dot{x}^\mathrm{c,C}&=f^\mathrm{c,C}(x^\mathrm{c,C}, u)
			\end{align}
			This approach in \eqref{eq:dae_ode_frenet} or \eqref{eq:dae_ode_frenet_final} (\emph{Lifted ODE Frenet}) results in redundant states in both CFs, which are coupled through the inputs and the initial state.
			We do not consider formulations based on the CCF ODE, since it involves the challenging bi-level problem to obtain FCF states			
	\section{Obstacle avoidance formulations}
	\label{section:obstacle_avoidance}
		Different formulations for obstacle avoidance constraints are used in NMPC and visualized in Fig.\ref{fig:obstacles}.
		We assume that rectangles represent the real vehicle shapes.
		\begin{figure}
			\begin{center}
				\def\svgwidth{.45\textwidth}
\begingroup%
  \makeatletter%
  \providecommand\color[2][]{%
    \errmessage{(Inkscape) Color is used for the text in Inkscape, but the package 'color.sty' is not loaded}%
    \renewcommand\color[2][]{}%
  }%
  \providecommand\transparent[1]{%
    \errmessage{(Inkscape) Transparency is used (non-zero) for the text in Inkscape, but the package 'transparent.sty' is not loaded}%
    \renewcommand\transparent[1]{}%
  }%
  \providecommand\rotatebox[2]{#2}%
  \newcommand*\fsize{\dimexpr\f@size pt\relax}%
  \newcommand*\lineheight[1]{\fontsize{\fsize}{#1\fsize}\selectfont}%
  \ifx\svgwidth\undefined%
    \setlength{\unitlength}{558.21467922bp}%
    \ifx\svgscale\undefined%
      \relax%
    \else%
      \setlength{\unitlength}{\unitlength * \real{\svgscale}}%
    \fi%
  \else%
    \setlength{\unitlength}{\svgwidth}%
  \fi%
  \global\let\svgwidth\undefined%
  \global\let\svgscale\undefined%
  \makeatother%
  \begin{picture}(1,0.42228315)%
    \lineheight{1}%
    \setlength\tabcolsep{0pt}%
    \put(0,0){\includegraphics[width=\unitlength,page=1]{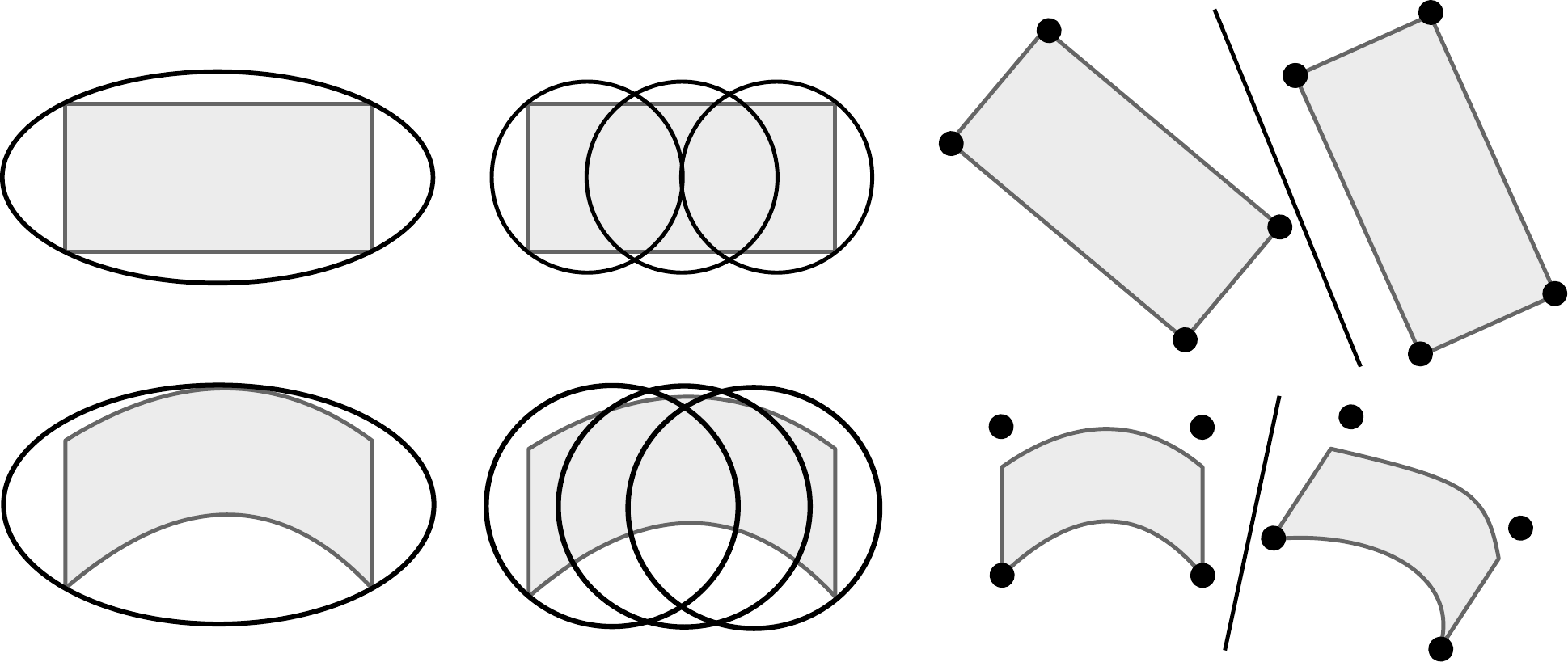}}%
    \put(0.00559788,0.37305628){\makebox(0,0)[lt]{\lineheight{1.25}\smash{\begin{tabular}[t]{l}a)\end{tabular}}}}%
    \put(0.30118305,0.37305628){\makebox(0,0)[lt]{\lineheight{1.25}\smash{\begin{tabular}[t]{l}b)\end{tabular}}}}%
    \put(0.58199955,0.37305628){\makebox(0,0)[lt]{\lineheight{1.25}\smash{\begin{tabular}[t]{l}c)\end{tabular}}}}%
    \put(0.00551038,0.1742562){\makebox(0,0)[lt]{\lineheight{1.25}\smash{\begin{tabular}[t]{l}d)\end{tabular}}}}%
    \put(0.30109553,0.1742562){\makebox(0,0)[lt]{\lineheight{1.25}\smash{\begin{tabular}[t]{l}e)\end{tabular}}}}%
    \put(0.58191203,0.1742562){\makebox(0,0)[lt]{\lineheight{1.25}\smash{\begin{tabular}[t]{l}f)\end{tabular}}}}%
    \put(0.78845777,0.19320695){\makebox(0,0)[lt]{\lineheight{1.25}\smash{\begin{tabular}[t]{l}$h^\theta$\end{tabular}}}}%
    \put(0.35851563,0.20664264){\makebox(0,0)[lt]{\lineheight{1.25}\smash{\begin{tabular}[t]{l}$n_\mathrm{circ}=3$\end{tabular}}}}%
    \put(0,0){\includegraphics[width=\unitlength,page=2]{obstacles_all.pdf}}%
  \end{picture}%
\endgroup%

				\caption{Schematic drawing of obstacle constraints. (a: ellipse CCF, b: covering circles CCF, c: separating hyperplanes CCF, d: ellipse FCF, e: covering circles FCF, f: separating hyperplanes FCF)}
				\label{fig:obstacles}
			\end{center}
		\end{figure}
		Often simple geometric covering shapes (circles \cite{Khorkov2021} or ellipses \cite{Brito2019, Nair2022a}) and related distance functions are used.
		Alternatively, covering polygons and restrictions on edges or vertices (hyperplanes) are formulated, \cite{Sathya2018,Brossette2017, Evens2022}.
		Furthermore, the boundaries can be deflected in order to integrate the obstacle into the boundary constraints \cite{Kloeser2020}. 
		The later approach is not within the scope of this work, due to the generally different formulations that for instance, include a combinatorial planner for choosing the passing side \cite{Reiter2021a}.
		We compare the formulation of obstacle avoidance constraints by an 
		\emph{ellipse} \cite{Brito2019}, \emph{covering circles} \cite{Wang2021, Khorkov2021} and \emph{separating hyperplanes} \cite{Brossette2017}.
		We also implemented a formulation introduced in \cite{Sathya2018}, which we refer to as \emph{set-vertices-exclusion}, but which poorly converged in our experiments.
		We assume a rectangular shape of the vehicles with the rear/front chassis length~$l_\mathrm{ch}=l_\mathrm{r,ch}+l_\mathrm{f,ch}$ related to the vehicle CG and chassis width~$w_\mathrm{ch}$.
		When we considering the over-approximated area to evaluate which constraint formulations is superior, the separating hyperplane do not increase the obstacle size, whereas circles and ellipses add additional unoccupied area.
		\subsection{Obstacle approximation by an ellipse}
		\label{sec:obstacle_ellipse}
		The distance between a circle and an ellipse can be computed explicitly, while computing the distance between two ellipses is more involved.
		Thus, we cover the ego car with a circle.
		The main axes~$a,b$ of an ellipse covering a rectangle are computed by ${a=\frac{1}{\sqrt{2}}(l_\mathrm{f,ch}+l_\mathrm{r,ch})}$ and $b=\frac{1}{\sqrt{2}}w_\mathrm{ch}$.
		Increasing this by the ego radius~$r^\mathrm{ego}$ leads to the extended ellipse matrix~$D=\mathrm{diag}(\left[a+r^\mathrm{ego},b+r^\mathrm{ego}\right])$.
		With the rotation matrix~$R(x^\mathrm{c,opp})\in \R^{2\times2}$ and a translation vector~$t(x^\mathrm{c,opp}) \in \R^{2}$ related to the orientation and position of the obstacle vehicle, we can formulate the collision avoidance constraint with the matrix~$\Sigma(x^\mathrm{c,opp})=R(x^\mathrm{c,opp}) D R(x^\mathrm{c,opp})^\top$ via the feasible set
		\begin{align}
		\label{eq:collision_ellipse}
		\begin{split}
			\free^\mathrm{ell}(x^\mathrm{c,opp}) = 
			\Bigl\{&x^\mathrm{c} \in \R^3\Big|
			\norm{x^\mathrm{c}-t(x^\mathrm{c,opp})}^2_{\Sigma^{-1}(x^\mathrm{c,opp})}
			\geq 1\Bigr\} .
		\end{split}
		\end{align}
		\subsection{Obstacle approximation by covering circles}
		\label{obstacle:circles}
		We use the union of a \emph{set of circles} to contain the set of the vehicle shape \cite{Wang2021, Ziegler2014, Khorkov2021}.
		For $l_\mathrm{ch}\geq w_\mathrm{ch}$, the number of covering circles~$n_\mathrm{circ}$ must be larger than $\lceil\frac{l_\mathrm{ch}}{w_\mathrm{ch}} \rceil$ and have a radius $r^{\{\mathrm{ego,opp}\}}$ of $w_\mathrm{ch} \frac{1}{\sqrt{2}}$.
		For each combination of the $n_\mathrm{circ,ego}$ and $n_\mathrm{circ,opp}$ circles a distance constraint must be satisfied, leading to $n_\mathrm{circ,ego} n_\mathrm{circ,opp}$ inequality constraints.
		The covering circle center points are computed according to \cite{Ziegler2014}, which gives us a function $p^i:\R^3\rightarrow \R^2$ for the circle center $i$ that computes the center positions $p_i=p^i(x^\mathrm{c,C})$ from the states $x^\mathrm{c,C}$.
		With $\Delta r= r^\mathrm{ego} + r^\mathrm{opp}$ and $x:=x^\mathrm{c,C}$, we write the free set as
		\begin{align}
		\label{eq:collision_circles}
		\begin{split}
		\free^\mathrm{circ}(x^\mathrm{opp}) &= 
		\biggl\{x \in \R^3\big|
		\norm{p^i(x)-p^j(x^\mathrm{opp})}_2 \geq \Delta r, \\
		&\textrm{for}~ 1 \leq i \leq n_\mathrm{circ, ego}, 1 \leq j \leq n_\mathrm{circ, opp} \biggr\}
		\end{split}
		\end{align}
		\subsection{Obstacle approximation by separating hyperplanes}
		When formulating collision avoidance with separating hyperplanes, we optimize for a feasible solution of the parameters $\theta\in \R^3$ of a hyperplane~$h^\theta(p)$.
		The parameterized hyperplane needs to separate all (four) vertices~$p_i^{\{\mathrm{ego,opp}\}}(x^\mathrm{c\{ego,opp\}})$ of either vehicle's bounding box.
		We write the feasible region using the related hyperplane existence problem with points~${\bar{p}^{\{.\}\top}=[p^{\{.\}\top} \quad 1]}$ as
		\begin{align}
			\begin{split}
				\label{eq:hyperplane_theta}
				&\free^\mathrm{hp}(x^\mathrm{opp})=\Bigl\{x \in \R^3, \theta \in \R^3\Big|\text{ } 
				\theta_1^2 + \theta_2^2=1,\\
				&\text{ }\theta^\top\bar{p}_i^{\mathrm{ego}}(x)\leq 0,\text{ } 
				\theta^\top\bar{p}_i^{\mathrm{opp}}(x^\mathrm{opp})\geq 0,\text{ } 
				\forall i=0,\ldots,3
				\Bigr\}.
			\end{split}
		\end{align}
		With the constraint $\theta_1^2 + \theta_2^2=1$ for the hyperplane parameters we avoid a degenerate solution.
		
	\section{NMPC Formulation}
	\label{section:approach}
		The aim of the NMPC is to plan a feasible trajectory of a vehicle to drive on a road with bounded curvature on a reference lane parallel to the center line and with a desired reference speed.
		Furthermore, the NMPC must avoid static and dynamic obstacles.
		As motivated in Sec.~\ref{sec:overview_lifting}, we use two variants of an FCF-based ODE to obtain Cartesian states, i.e. the \emph{direct elimination} and \emph{lifted ODE} formulation and compare it to the \emph{conventional} formulation with over-approximation, such as shown in \cite{Rosolia2017, Ayoub2022}.		
		First, we define the costs and constraints.
		\subsubsection{General costs \& constraints}
		Some constraints are unrelated to the CF, such as the lower and upper bounds for states $x^\mathrm{\neg c}$ and inputs $u$.
		For control costs ~$u^\top R u$, we use the positive semi-definite weight matrix~$R\in \R^{2\times2}$.
		\subsubsection{FCF related costs \& constraints}
		State costs are related to FCF states, since there is no practical advantage of including CCF state costs.
		A cost related to a desired reference path parallel to the road center line is accounted for by a square penalty of the deviation of the Frenet lateral coordinate~$n$ to its reference~$n_{\mathrm{ref}}$.
		For a reference speed~$v_\mathrm{ref}$, a square penalty with positive weight~$w_s$, as well as a penalty on precomputed longitudinal reference positions~$s_{\mathrm{ref},i}=\hat{s}_{0}+i \Delta t v_\mathrm{ref}$ is used, with the measured state~$\hat{s}_{0}$ and sampling time $\Delta t$. 
		Since we assume a road with constant width, boundary constraints simplify in the FCF to box constraints for an upper bound~$\ub{n}$ and a lower bound $\lb{n}$.
		\subsubsection{CCF related costs \& constraints}
		We use the collision avoidance formulations, which could be one out of $\mathcal{O}=\{\mathrm{ell},\mathrm{circ},\mathrm{hp}\}$ in the CCF, thus have the constraint ${x^\mathrm{c,C} \in \free^{\{\mathrm{ell},\mathrm{circ},\mathrm{hp}\}}}$.
		FCF costs are defined via the positive weight matrix~$Q=\text{diag}(q)$ with the weight vector $q\in\R^5$ and the reference states in~$x^\mathrm{F}_\mathrm{ref}$.
		We use a terminal cost~$Q_N=\text{diag}(q_N)$ with the weight vector $q_N\in\R^5$ after~$N$ discrete time steps. 
		\subsubsection{Direct elimination NMPC formulation}
		With the direct formulation, we can directly use the inverse transformation~${x^\mathrm{c,C}=\frenetTrans^{-1}(x^\mathrm{c,F})}$ to eliminate the Cartesian states in the constraint formulation.
		Thus we have less states, but "more" nonlinear constraints.
		We discretize the continuous trajectory with~${N-1}$ control intervals and use direct multiple shooting \cite{Bock1984} with one step of an RK4 integration function~$\Phi^\mathrm{F}(x^\mathrm{F},u, \Delta t)$ for the ODE in \eqref{eq:model_frenet_full} and the NLP formulation 
		
		\begin{small}
			\begin{mini}
				{\begin{subarray}{c}
						x^\mathrm{F}_0, \ldots, x^\mathrm{F}_N,\\
						\control_0, \ldots, \control_{N-1}\\
						\theta_1, \ldots, \theta_{n_\mathrm{opp}}
				\end{subarray}}			
				{\sum_{k=0}^{N-1} \norm{\control_k}_{R}^2 + \norm{x^\mathrm{F}_k-x^\mathrm{F}_{\mathrm{ref},k}}_{Q}^2 + \norm{x^\mathrm{F}_N-x^\mathrm{F}_{\mathrm{ref},N}}_{Q_N}^2}
				{\label{eq:MPC_2}} 
				{} 
				\addConstraint{x^\mathrm{F}_0}{= \bar{x}^\mathrm{F}_0}{},
				\addConstraint{x^\mathrm{F}_{i+1}}{= \Phi^\mathrm{F}(x^\mathrm{F}_i,u_i, \Delta t),}{ i=0,\ldots,N-1},
				\addConstraint{\lb{u}}{\leq \control_i\leq \ub{u} ,}{i=0,\ldots,N-1},
				\addConstraint{\lb{x}^\mathrm{F} }{\leq x^\mathrm{F}_{i} \leq\ub{x}^\mathrm{F}, }{i=0,\ldots,N},
				\addConstraint{\lb{x}^\mathrm{c,C} }{\leq \frenetTrans^{-1}(x^\mathrm{c,F}) \leq\ub{x}^\mathrm{c,C}, }{i=0,\ldots,N},
				\addConstraint{\lb{a}^\mathrm{lat} }{\leq a^\mathrm{F}_\mathrm{lat}(x_i) \leq\ub{a}^\mathrm{lat},}{i=0,\ldots,N},
				\addConstraint{v_N }{\leq \bar{v}_N}{},
				\addConstraint{\frenetTrans^{-1}(x^\mathrm{c,F})}{\in \free(x^\mathrm{c,opp,j}_i, \theta_j),}{i=0,\ldots,N-1},
				\addConstraint{}{}{j=1,\ldots,n_\mathrm{opp}}.
			\end{mini}
		\end{small}
		
		Decision variables $\theta_1,\ldots,\theta_{n_\mathrm{opp}}$, where ${\theta_j = [\theta_j^0,\ldots,\theta_j^{N}]\in\R^{3\times N}}$ are only used for the separating hyperplanes formulation.
		\subsubsection{Lifted ODE NMPC formulation}
		In this formulation we use the extended state ~$x^\mathrm{d}=[x^{\mathrm{F}\top} \quad x^{\mathrm{c,C}\top }]^\top$
		and the extended ODE \eqref{eq:dae_ode_frenet_final}.
		The additional states increase the size of the state-space to eight in our case, where three states stem from either CF and additional two states are CF independent states. 
		In this formulation we use the RK4 integration function~$\Phi^\mathrm{d}(x^\mathrm{d},u, \Delta t)$ of dynamics \eqref{eq:dae_ode_frenet_final}.
		We can write the final NLP for the \emph{lifted ODE} formulation as
			\begin{mini}
				{\begin{subarray}{c}
						x^\mathrm{d}_0, \ldots, x^\mathrm{d}_N,\\
						\control_0, \ldots, \control_{N\!-\!1}\\
						\theta_1, \ldots, \theta_{n_\mathrm{opp}}
				\end{subarray}}			
				{\sum_{k=0}^{N-1} \norm{\control_k}_{R}^2 \!+\! \norm{x^\mathrm{F}_k-x^\mathrm{F}_{\mathrm{ref},k}}_{Q}^2 \!+\! \norm{x^\mathrm{F}_N-x^\mathrm{F}_{\mathrm{ref},N}}_{Q_N}^2}
				{\label{eq:MPC}} 
				{} 
				\addConstraint{x^\mathrm{d}_0}{= \bar{x}^\mathrm{d}_0}{},
				\addConstraint{x^\mathrm{d}_{i+1}}{= \Phi^\mathrm{d}(x^\mathrm{d}_i,u_i, \Delta t),}{ i=0,\ldots,N-1},
				\addConstraint{\lb{u}}{\leq \control_i\leq \ub{u} ,}{i=0,\ldots,N-1},
				\addConstraint{\lb{x}^\mathrm{d} }{\leq x^\mathrm{d}_{i} \leq\ub{x}^\mathrm{d}, }{i=0,\ldots,N},
				\addConstraint{\lb{a}^\mathrm{lat} }{\leq a_\mathrm{lat}(x^\mathrm{d}_i) \leq\ub{a}^\mathrm{lat},}{i=0,\ldots,N},
				\addConstraint{v_N }{\leq \bar{v}_N}{},
				\addConstraint{x^\mathrm{c,C}_i}{\in \free(x^\mathrm{c,opp,j}_i, \theta_j),}{i=0,\ldots,N-1},
				\addConstraint{}{}{j=1,\ldots,n_\mathrm{opp}}.
			\end{mini}
	\section{Numerical experiments}
	\label{section:results}
		In order to evaluate the performance of the proposed approach, we simulate two randomized scenarios that constitute three non-ego vehicles in front of the ego vehicle with a lower cruise speed.
		The scenario is simulated for 20 seconds, where usually three overtakes are possible.
		In total, 500 full simulation runs are evaluated for each scenario type.
		We record the solution times of the NMPC and the final driven distance after the simulation ended, which we take as a performance indicator.
		We use different types of obstacles, particularly long ones in the dimensions of a truck (truck-sized), as well as short ones resembling normal cars (car-sized).
		We make several simplifications in order to avoid performance influences of sources unrelated to our formulation.
		Firstly, there is no model-plant mismatch, i.e., the simulation framework and the NLP use the same kinematic vehicle model and discretization. 
		Secondly, the ego NMPC has access to the planned trajectories of the other vehicles in order to avoid an influence of prediction errors.
		Finally, we model non-ego participants to be non-interactive.
		They aim at driving along a reference line parallel to the center line.
		The simulations were run on an Alienware~m-15 Notebook with an Intel Core i7-8550 CPU (1.8 GHz). 
		The parameters for the environment and the NMPC are shown in Tab.~\ref{table:parameters_vehicles} and Tab.~\ref{table:parameters_mpc}, respectively.
		\begin{table}
			\centering
			\ra{1.1}
			\begin{tabular}{@{}llll@{}}
				\addlinespace
				\toprule
				Module & Name & Variable & Value  \\
				\midrule
				Road & road bounds$^2$ & $\lb{n},\ub{n}$& $\pm 8.5$\\
				& curvature$^1$ & $\kappa$ & $[\shortminus0.05, 0.05]$\\
				& wind speed & $v_\mathrm{wind}$ & $20$\\
				& wind direction & $\varphi_\mathrm{wind}$ & $0$\\
				\midrule
				Ego vehicle & length wheelbase & $l_\mathrm{r}, l_\mathrm{f}$ & $1.7$  \\
				& length chassis & $l_\mathrm{r,ch}, l_\mathrm{f,ch}$& $2$ \\
				& width chassis& $w_\mathrm{ch}$& $1.9$\\
				& mass & m & 1160\\
				& lateral acc. bound & $\lb{a}_\mathrm{lat},\ub{a}_\mathrm{lat}$ & $\pm 5$\\
				& input bounds & $\lb{u}, \ub{u}$ & $\pm[10\mathrm{kN}, 0.39]$ \\
				& velocity bound & $\ub{v}$ & $40$\\
				& steering angle bound & $\lb{\delta},\ub{\delta}$ & $\pm0.3 $ \\
				& starting position$^1$& $x^\mathrm{c,F}_0$& $[0,\shortminus5,0]$- \\
				& & & $ [0,5,0]$ \\
				& reference speed & $v_\mathrm{ref}$& $40$\\
				\midrule
				Opp. vehicles & length wheelbase$^2$ & $l_\mathrm{r}, l_\mathrm{f}$ & $10$  \\
				& length chassis$^2$ & $l_\mathrm{r,ch}, l_\mathrm{f,ch}$& $13$ \\
				& width chassis$^2$ & $w_\mathrm{ch}$& $4$\\
				& mass$^2$ & m & $3000$\\
				& input bounds$^2$ & $\ub{u}$ & $[30\mathrm{kN}, 0.39]$ \\
				& input bounds$^2$ & $\lb{u}$ & $[\shortminus45\mathrm{kN}, \shortminus0.39]$ \\
				& starting position$^1$ $i$& $s^i_0$& $50(i+1)$ \\
				& &$n^i_0$ & $ [\shortminus 5,5]$ \\
				& reference speed & $v_\mathrm{ref}$& $15$\\
				\bottomrule
			\end{tabular}
			\caption{Environment parameters.$^1$ Randomized with uniform distribution.$^2$ Parameters only differ in long vehicle scenario. The parameters are equal for all vehicles, if not noted explicitly. We use SI units, if not stated explicitly.
			}
			\label{table:parameters_vehicles}
		\end{table}
		We use the NLP solver \texttt{acados}~\cite{Verschueren2021} with \texttt{HPIPM}~\cite{Frison2020a}, RTI iterations and a partial condensing horizon of $\frac{N}{2}$.
		\begin{table}
			\centering
			\ra{1.2}
			\begin{tabular}{@{}lll@{}}
				\toprule
				Name & Variable & Value  \\
				\midrule
				nodes / disc. time & $N$/ $\Delta t$  & $40$/ $0.1$\\
				terminal velocity & $\ub{v}_N$ & $15$ \\
				state weights & $q$ & $[1, 500, 10^3, 10^3, 10^4]\Delta t$\\
				terminal state weights  & $q_N$ & $[10, 90, 100, 10, 10]$\\
				control weights  & $R$ & diag($[10^{-3}, 2\cdot 10^6])\Delta t$\\
				\bottomrule
			\end{tabular}
			\caption{Parameters for MPC in SI units.
			}
			\label{table:parameters_mpc}
		\end{table}
		We use obstacle constraint formulations of Sec.~\ref{section:obstacle_avoidance}.
		Besides the different obstacle dimensions, the proposed NMPC formulations \emph{convetional}, \emph{direct elimination} and \emph{lifted ODE} were evaluated with the different obstacle formulations of Sec.~\ref{section:obstacle_avoidance}.
		We use the ellipse ("EL"), the $n$ covering circles ("C$n$") and the separating hyperplane ("HP") formulation.
		In Fig~\ref{fig:eval_short} the computation times and the maximum achieved progress of the randomized scenarios are shown. 
		Clearly, final progress after overtaking in the truck-sized scenario is increased in the proposed formulation significantly due to the more accurate representation of the obstacle shape.
		For car-sized vehicles the extended states do not yield a prominent advantage, since here the Frenet transformation does not deform the obstacles vastly.
		The maximum progress is nearly equal for both proposed approaches, since the obstacle constraint formulations based on Cartesian states is equal.
		A striking difference between the two proposed formulations can be seen in the computation times, shown detailed in Tab.~\ref{table:timing_comparison}.
		While the \emph{lifted ODE} formulation even decreases the average computation time for nearly all obstacle formulations, the \emph{direct elimination} formulation increases the computation time by around 30\%.
		Remarkably, the ellipsoidal obstacle formulation in the proposed lifted ODE formulation outperforms all other obstacle formulations in both, the computation time, as well as the performance measured in the average progress after overtaking, which highlights the advantage of the formulation.
		Contrary to our expectations, the separating hyperplane formulation shows weaker performance in computation time and average progress.
		In theory, separating hyperplanes should be more accurate in capturing the obstacle shape, nevertheless, due to the disadvantageous linearizations within the SQP iterations, the shape is not captured well.
		This might be mainly due to the nonconvex and nonlinear constraint in \eqref{eq:hyperplane_theta}.
		\begin{figure}
			\label{fig:eval_short}
			\begin{multicols}{2}
				\centering
				\includegraphics[scale=0.8]{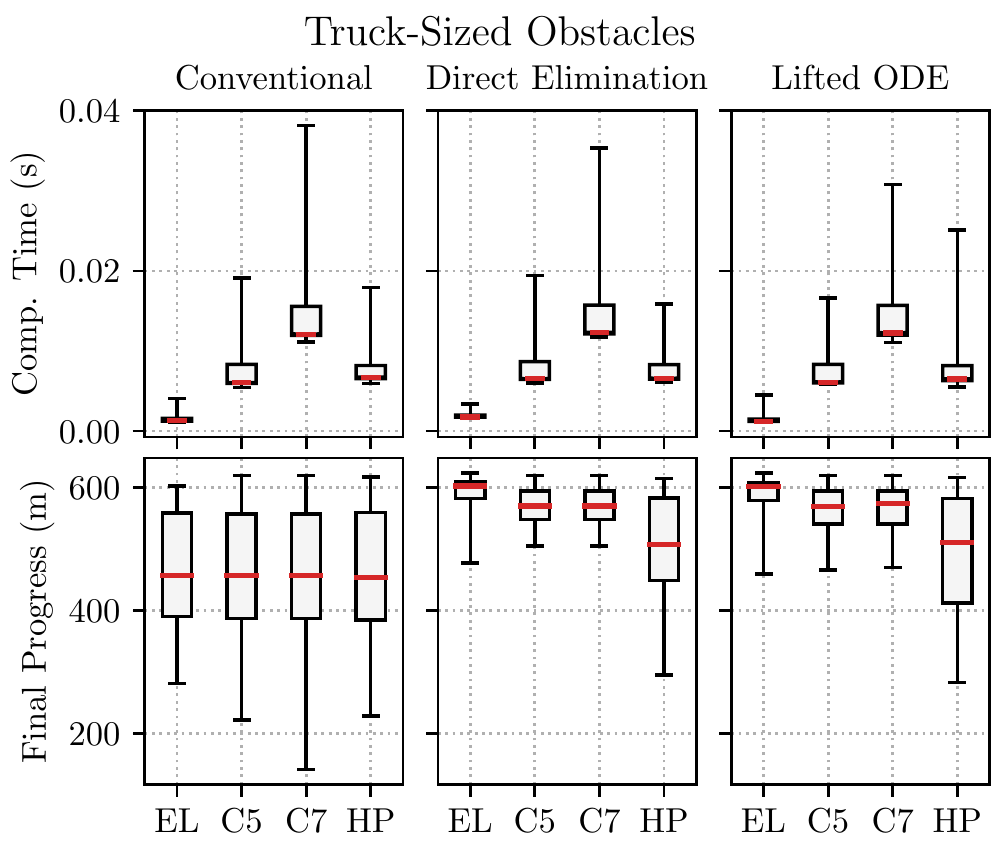}
				\includegraphics[scale=0.8]{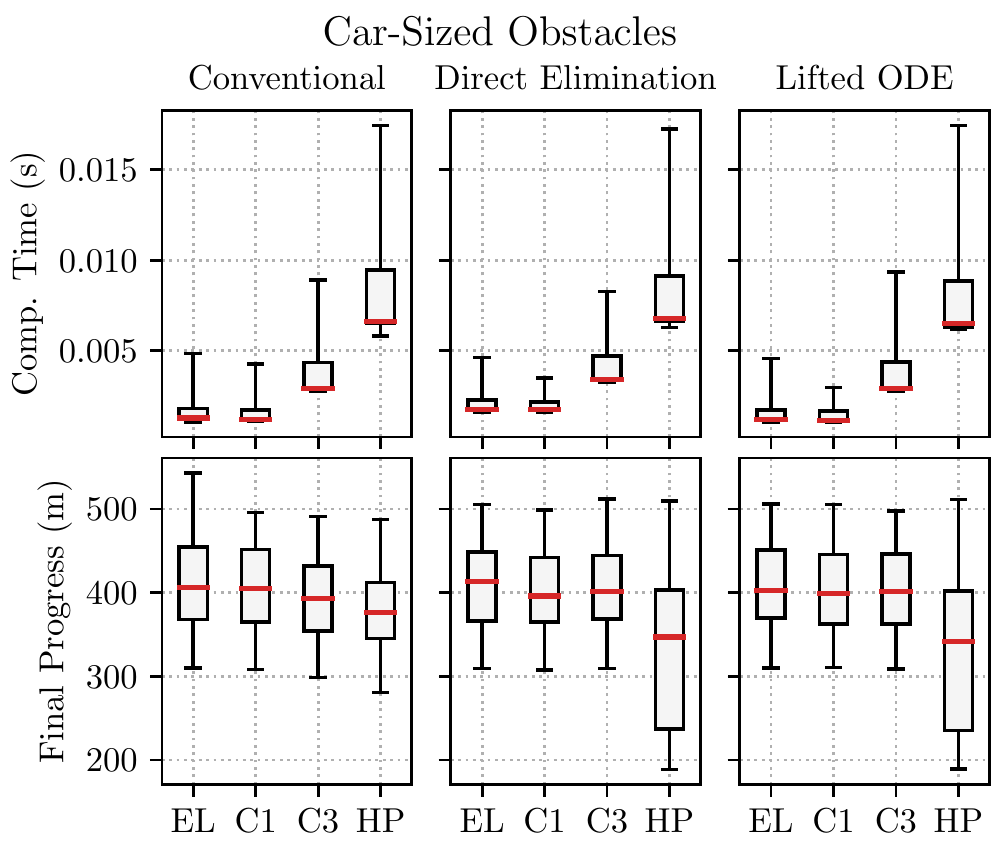}
			\end{multicols}
			\caption{Box-plot comparison of the MPC solution timings for each real-time iteration and the final progress after 20 seconds for different obstacle formulations for truck- and car-sized vehicles.}
		\end{figure}
		\begin{table}
			\centering
			\ra{1.1}
			\begin{tabular}{@{}l|l|ll|ll@{}}
				\toprule
				\multicolumn{6}{c}{Computation times (ms) for truck-sized obstacles}\\
				 & Conventional & \multicolumn{2}{l}{Direct Elimination}  & \multicolumn{2}{c}{Lifted ODE}  \\
				\midrule
				EL& $       1.5\pm        0.4$ & $       1.9\pm        0.2$ &$      28.9\%$ & $       1.4\pm        0.3$ & $      -6.6\%$\\
				C5& $       7.2\pm        1.9$ & $       7.6\pm        1.7$ &$       5.5\%$ & $       7.2\pm        1.8$ & $      -0.0\%$\\
				C7& $      14.0\pm        3.2$ & $      14.0\pm        2.8$ &$      -0.1\%$ & $      13.9\pm        2.9$ & $      -0.4\%$\\
				HP& $       7.5\pm        1.5$ & $       7.5\pm        1.5$ &$      -0.1\%$ & $       7.4\pm        1.7$ & $      -1.6\%$\\
				\toprule
				\multicolumn{6}{c}{car-sized obstacles}\\
				\midrule
				EL& $       1.5\pm        0.5$ & $       2.0\pm        0.4$ &$      29.6\%$& $       1.4\pm        0.4$ & $      -5.7\%$\\
				C1& $       1.4\pm        0.4$ & $       1.9\pm        0.4$ &$      34.0\%$& $       1.4\pm        0.4$ & $      -3.5\%$\\
				C3& $       3.6\pm        1.1$ & $       4.0\pm        1.0$ &$      12.4\%$& $       3.6\pm        1.1$ & $       0.6\%$\\
				HP& $       8.0\pm        2.3$ & $       7.9\pm        1.9$ &$      -0.6\%$& $       7.7\pm        2.0$ & $      -4.0\%$\\
				\bottomrule
			\end{tabular}
			\caption{Mean and standard deviation of computation times for different scenarios, obstacle formulations and lifting formulations. Additionally, the difference in percent to the conventional formulation is given.}
			\label{table:timing_comparison}
		\end{table}	
		A rendered resulting simulation in both CFs can be found on the website \url{https://rudolfreiter.github.io/obstacle_avoidance/}
		\label{section:setup}
	\section{Conclusions}
	We have presented two novel FCF-based formulations of MPC for vehicles that include states of the CCF in order to gain numerical advantages.
	Simulated evaluations and several wide-spread obstacle constraint formulations show that the proposed approaches are capable of representing the obstacle shapes more suitably and that with the \emph{lifted ODE} formulation even the computation time was decreased.
	Furthermore, our evaluations show that an ellipsoidal obstacle representation outperforms all other obstacle formulations in computation time.
	In conclusion, we state that the combination of the ellipsoidal obstacle constraint formulation and the \emph{lifted ODE} formulation yields superior results in all categories.
	\section*{ACKNOWLEDGMENT}
	This research was supported by DFG via Research Unit FOR 2401 and project 424107692 and by the EU via ELO-X 953348.
	\bibliographystyle{IEEEtran}
 	\bibliography{syscop,templib}
\end{document}